%% file: arxiv.tex
\RequirePackage{fix-cm}
\documentclass[12pt,a4paper,notitlepage]{article}
\pdfoutput=1 
\usepackage[hyphens]{url}
\usepackage{hyperref}

\usepackage{amsmath,mathtools,etoolbox,enumerate,xspace,amsthm,amssymb,amsfonts}
\usepackage{microtype}
\mathtoolsset{mathic,centercolon}

\usepackage{xcolor}
\hypersetup{
  colorlinks,
  linkcolor={red!70!black},
  citecolor={green!50!black},
  urlcolor={blue!50!black},
  pdfauthor={Floris van Doorn},
  unicode,
  pdftitle={Designing a general library for convolutions}}
\usepackage[numbers]{natbib}

\newtheorem{theorem}{Theorem}[section]

\newtheorem{proposition}[theorem]{Proposition}
\newtheorem{lemma}[theorem]{Lemma}
\newtheorem{corollary}[theorem]{Corollary}

\theoremstyle{definition}
\newtheorem{definition}[theorem]{Definition}

\theoremstyle{remark}

\usepackage{ucs}
\usepackage[utf8x]{inputenc}
\usepackage[LGR, T1]{fontenc}
\usepackage{textcomp}
\usepackage{textgreek}
\usepackage{tikz-cd}
\usepackage{mathrsfs}
\usepackage{etoolbox}

\usepackage[english]{babel}

\usepackage[capitalize]{cleveref}
\usepackage{xspace}

\usepackage{fontawesome} 
\usepackage{mathtools} 
\usepackage{upgreek} 

\definecolor{keywordcolor}{rgb}{0.7, 0.1, 0.1}   
\definecolor{commentcolor}{rgb}{0.4, 0.4, 0.4}   
\definecolor{symbolcolor}{rgb}{0.0, 0.1, 0.6}    
\definecolor{sortcolor}{rgb}{0.1, 0.5, 0.1}      
\definecolor{errorcolor}{rgb}{1, 0, 0}           
\definecolor{stringcolor}{rgb}{0.5, 0.3, 0.2}    

\usepackage{listings}

\usepackage{flushend}

\input{macros_env.tex}

\title{Formalising the $h$-principle and sphere eversion}

\author{Patrick Massot \and Floris van Doorn \and Oliver Nash}
\date{}

\begin{document}
\maketitle

\input{abstract.tex}

\input{body.tex}

\section*{Acknowledgements}
\input{acks.tex}

\Urlmuskip=0mu plus 1mu 
\bibliographystyle{amsalphaurl}
\bibliography{eversion}
\end{document}

%% file: macros_env.tex
\newcommand{\mathlib}{\texttt{mathlib}\xspace}

\newcommand{\separated}{separated\xspace}

\newcommand{\nhds}[1]{\mathcal{N}_{#1}}

\newcommand{\NN}{\mathbb{N}}

\newcommand{\ZZ}{\mathbb{Z}}

\newcommand{\RR}{\mathbb{R}}

\newcommand{\CC}{\mathbb{C}}

\renewcommand{\SS}{\mathbb{S}}
\renewcommand{\S}{\mathbb{S}^1}

\newcommand{\Loop}{\mathcal{L}}

\long\def\C{\mathcal{C}}

\DeclareMathOperator{\id}{id}

\DeclareMathOperator{\pr}{pr}

\newcommand{\e}{\varepsilon}

\renewcommand{\subset}{\subseteq}

\newcommand{\link}{\ensuremath{{}^\text{\faExternalLink}}}

\newcommand{\Rel}{\mathcal{R}}

\newcommand{\F}{\mathcal{F}}

\newcommand{\PS}[1]{\mathcal{P}(#1)}

\newcommand{\Fi}[1]{\mathcal{F}(#1)}

\newcommand{\forallf}[2]{\forall^f #1 ∈ #2}

\newcommand{\Corr}[3]{\mathcal{T}_{#1}^{#2}#3}

\newcommand{\Rem}[3]{R_{#1}^{#2}#3}

\newcommand{\avg}[1]{\overline{#1}}

\DeclareMathOperator{\Conn}{Conn}
\newcommand{\conn}[2]{\Conn_{#2}(#1)}

\newcommand{\co}{\!:}
\newcommand{\rst}[2]{#1|_{#2}}

\newcommand{\Upd}[3]{\Upsilon_{#1}\ifstrempty{#2}{}{\left(#2, #3\right)}}

\DeclareMathOperator{\Span}{Span}

\newcommand{\rot}[2]{\mathrm{Rot}_{#1, #2}}

\renewcommand{\L}[2]{L(#1, #2)}

\newcommand{\R}{\mathbb{R}}
\newcommand{\mllink}[1]{
\href{https://github.com/leanprover-community/mathlib/blob/8b8cd99cfd2c6ef5caffe15a358f56609e494e8e/src/#1}{\link}
}
\newcommand{\selink}[1]{
\href{https://github.com/leanprover-community/sphere-eversion/blob/cpp2023/src/#1}{\link}
}
\DeclareMathOperator{\CH}{Conv} 
\DeclareMathOperator{\im}{im} 

%% file: abstract.tex
\begin{abstract}
  In differential topology and geometry, the h-principle is a property enjoyed by certain
  construction problems. Roughly speaking, it states that the only
  obstructions to the existence of a solution come from algebraic topology.

  We describe a formalisation in Lean of the local h-principle for first-order, open,
  ample partial differential relations. This is a significant result in
  differential topology, originally proven by Gromov in 1973 as part of
  his sweeping effort which greatly generalised many previous flexibility
  results in topology and geometry. In particular it reproves Smale's celebrated
  sphere eversion theorem, a visually striking and counter-intuitive construction.
  Our formalisation uses Theillière's implementation of convex integration from 2018.

  This paper is the first part of the \textsf{sphere eversion project}, aiming to
  formalise the \emph{global} version of the h-principle for open and ample first
  order differential relations, for maps between smooth manifolds.
  Our current local version for vector spaces is the main ingredient of this
  proof, and is sufficient to prove the titular corollary of the project.
  From a broader perspective, the goal of this project is to show that one can formalise
  advanced mathematics with a strongly geometric flavour and not only algebraically-flavoured
  mathematics.
\end{abstract}

%% file: body.tex
\section{Introduction}%
\label{sec:introduction}

\subsection{Formal proofs and geometric intuition}%
\label{sub:formal_proofs_and_geometric_intuition}

Mathematical arguments that rely on human geometric intuition can be seen as
challenges to formalisation. Logic, discrete mathematics and algebra have
certainly been formalised much more often. Two notable exceptions are
an elementary proof of Jordan's curve theorem in \cite{Hales07} and
the Poincaré--Bendixson theorem in \cite{Immler20}. The goal of our
project is to make a strong case that even differential topology
can be formalised, including Smale's sphere eversion theorem.

The context of this theorem is Gromov's $h$-principle. He says that a
geometric construction problem satisfies the $h$-principle, or is flexible, if
the only obstructions to the existence of a solution come from algebraic
topology (the letter $h$ stands for ``homotopy'').

\subsection{A toy example}%
\label{sub:a_toy_example}

The easiest example of a flexible construction problem which is not
trivial and which is algebraically obstructed is the deformation of immersions of
circles into planes. Let $f_0$ and $f_1$ be two maps from $\SS^1$ to
$\mathbb{R}^2$ that are immersions. Since $\SS^1$ has dimension one,
this means that both derivatives $f_0'$ and $f_1'$ are nowhere vanishing maps
from $\SS^1$ to $\mathbb{R}^2$. The geometric object we want to
construct is a (smooth\footnote{In this paper, smooth always means infinitely
differentiable.}) homotopy of immersions from $f_0$ to $f_1$, i.e. a smooth map
$F \co \SS^1 \times [0, 1] \to \mathbb{R}^2$ such that
$\rst{F}{\SS^1 \times \{0\}} = f_0$,
$\rst{F}{\SS^1 \times \{1\}} = f_1$, and each $f_p :=
\rst{F}{\SS^1 \times \{p\}}$ is an immersion. If such a homotopy exists
then, $(t, p) \mapsto f_p'(t)$ is a homotopy from $f_0'$ to $f_1'$ among maps
from $\SS^1$ to $\mathbb{R}^2 \setminus \{0\}$. Such maps have a well
defined winding number $w(f'_i) \in \mathbb{Z}$ around the origin, the degree
of the normalised map $f'_i/\|f'_i\| \co \SS^1 \to \SS^1$. So
$w(f_0') = w(f_1')$ is a necessary condition for the existence of $F$, which
comes from algebraic topology.  The Whitney--Graustein theorem
states that this necessary condition is also sufficient.

The lesson to take away from the example above is that a necessary condition
imposed by algebraic topology may turn out to be sufficient.
Indeed the (one-dimensional)
Hopf degree theorem ensures that, provided $w(f'_0) = w(f'_1)$, there
exists a homotopy $g_p$ of nowhere vanishing maps relating $f'_0$ and
$f'_1$. We also know from the topology of $\mathbb{R}^2$ that $f_0$ and $f_1$ are
homotopic, e.g. via the straight-line homotopy
$p \mapsto f_p = (1-p)f_0 + pf_1$. But a priori there is no relation
between $g_p$ and the derivative of $f_p$ for $p \notin \{0, 1\}$.
So we can restate the crucial part of the
Whitney--Graustein theorem as: there is a homotopy of immersions from
$f_0$ to $f_1$ as soon as there is (a homotopy from $f_0$ to $f_1$) and
a homotopy from $f_0'$ to $f_1'$ among nowhere vanishing maps. The
parentheses in the previous sentence indicate that this condition is
always satisfied, but it is important to keep in mind for
generalisations. Gromov says that such a homotopy of uncoupled pairs
$(f, g)$ is a formal solution of the original problem. Note an unfortunate
terminology clash here: the word formal isn't used in the same sense as in
``formalised mathematics''.

\subsection{The $h$-principle}%
\label{sub:the_h_principle}

We now generalise this discussion of formal solutions
to the case of maps between two vector spaces $E$, $F$ with totally arbitrary
pointwise constraints on their derivatives.
A first-order partial differential relation for functions $E\to F$ is any subset
\[
  \Rel\subset E\times F\times \L{E}{F},
\]
where $\L{E}{F}$ denotes the linear maps from $E$ to $F$.
A function $f:E\to F$ satisfies $\Rel$ at a point $x$ in $E$ if
\[
  (x,f(x),Df(x))\in\Rel,
\]
where $Df(x)$ is the derivative of $f$ at $x$.

We say that a pair $(f,\varphi) : E \to F \times \L{E}{F}$ is a
\emph{formal solution} if for all $x\in E$ we have $(x,f(x),\varphi(x))\in\Rel$.
Any formal solution that is obtained from a function $f$
(i.e. $\varphi=Df$) is called \emph{holonomic},
showing that $f$ satisfies $\Rel$ everywhere. The product $E\times F\times \L{E}{F}$
is called the space of $1$-jets of maps from $E$ to $F$ and denoted by $J^1(E, F)$.

We say that $\Rel$ satisfies the \emph{$h$-principle} if any formal
solution can be deformed into a holonomic one inside $\Rel$.

Not all relations satisfy the $h$-principle. For instance the relation of
immersions of circles into a line rather than a plane fails to satisfy it.
A key insight of Gromov was to identify a \emph{geometric} condition on
the relation which ensures that the $h$-principle holds. In the special case of immersions
of circles into planes, at each point $x$ the derivative $f'(x)$ is required to
belong to the complement of $\{0\} \subseteq \RR^2$ and checking Gromov's
condition boils down to the fact that this complement is open and
\emph{ample}. A set in a real affine space is called ample if
all its connected components have the full space as their convex hull.
Note how this condition fails in the case of immersions of circles into lines:
in that case the complement of the origin has two connected components,
neither of whose convex hulls is the whole line.

Returning to the general case with $\Rel\subset J^1(E, F)$, one says
that $\Rel$ is ample if, for every $(x, y, \varphi) ∈ J^1(E, F)$ and every
hyperplane $H$ in $E$, the set:
\[
  \{\psi \in \L{E}{F} \;\left|\; \rst{\psi}{H} = \rst{\varphi}{H}\text{ and } (x, y, \psi) \in \Rel\right.\}
\]
is ample in the affine space of linear maps
that coincide with $\varphi$ on $H$. This set is called
the slice of $\Rel$ associated to $(x, y, \varphi)$ and $H$. In fact,
it is more convenient to use a variant of this definition which assumes more data than just $H$,
and so obtain slices that are subsets of $F$ rather than $\L{E}{F}$.
Instead of considering hyperplanes in $E$, we consider what we call \emph{dual pairs}.
This is a pair $p = (\pi, v) \in E^* \times E$ such that $\pi(v) = 1$. Using $p$, we
define the update map $\Upd{p}{}{} \co \L{E}{F} \to F \to \L{E}{F}$ sending
$(\varphi, w)$ to the linear map $\Upd{p}\varphi w$ which coincides with
$\varphi$ on $\ker \pi$ and sends $v$ to $w$.

The hyperplane associated to $p$ is $\ker \pi$ and the corresponding slice is
affine-isomorphic to
\[
  \Rel((x, y, \varphi), p) :=
  \left\{w \in F \;\left|\; \left(x, y, \Upd{p}\varphi w\right) \in \Rel\right.\right\}.
\]
Before stating the main theorem, we note that what we called the $h$-principle
above is only the weakest possible variation on this theme. One can add
constraints along various sets, proximity constraints, and parameters.
In the statement below, $P$ is the parameter space, it would be one-dimensional
in the example of homotopies of immersions. A family of formal solutions of
$\Rel$ parametrised by $P$ is a smooth map
$\F : P \times E \to F \times \L{E}{F}$ such that each $\F_p := \F(p, \cdot)$
is a formal solution.

\begin{theorem}[Local version of Gromov's $h$-principle for open and ample first order relations]
  \label{thm:gromov}
 Let $E$, $F$ and $P$ be finite dimensional real normed vector spaces.
 Let $\Rel$ be an open and ample set in $J^1(E, F)$. Let $C$ and $K$ be sets
 in $P \times E$ such that $C$ is closed and $K$ is compact.
 Let $\F_0$ be a family of formal solutions of $\Rel$ parametrised by $P$ such
 that $(\F_0)_p$ is holonomic at $x$ for all $(p, x)$ near $C$.

 For every positive real number $\e$, there exists a family of formal solutions
 $\F$ parametrised by $\RR \times P$ such that:
 \begin{itemize}
   \item $∀ p\, x, \F((0, p), x) = \F_0(p, x)$,
   \item $∀ (p, x) \text{ near } C, ∀ t, \F((t, p), x) = \F_0(p, x)$,
   \item $∀ t\, p\, x, \left\|\pr_F\big(\F((t, p), x) - \F_0(p, x)\big)\right\| ≤ \e$,
   \item $∀ (p, x) \text{ near }  K, \F_{(1, p)} \text{ is holonomic at } x$.
 \end{itemize}
 where $\pr_F : F \times \L{E}{F} \to F$ is the projection to the first factor.
\end{theorem}

This paper describes our formalisation\selink{local/parametric_h_principle.lean\#L296} of the above theorem.
We also explain how to apply it to obtain Smale's theorem whose formalised
statement follows (the notation \lstinline{I} refers to the unit interval in $\R$
and we omit the declaration of the fact that \lstinline{E} is a 3-dimensional
real vector space equipped with an inner product).\selink{local/sphere_eversion.lean\#L426}

\begin{lstlisting}
theorem sphere_eversion_of_loc :
  ∃ f : ℝ → E → E,
    (𝒞 ∞ ↿f) ∧
    (∀ x ∈ 𝕊², f 0 x = x) ∧
    (∀ x ∈ 𝕊², f 1 x = -x) ∧
    ∀ t ∈ I, sphere_immersion (f t)
\end{lstlisting}

\subsection{Convex integration}%
\label{sub:convex_integration}

The proof of \Cref{thm:gromov} uses convex integration, a technique
invented by Gromov c.1970, inspired by the $C^1$ isometric
embedding results of Nash and the original proof of flexibility of immersions.
This term is pretty vague however, and there are several different
implementations. The most recent, and by far the most efficient,
is Mélanie Theillière's corrugation process from \cite{Theilliere22}.
We opted for this implementation and describe it briefly below.

Let $f$ be a smooth map from $\mathbb{R}^n$ to $\mathbb{R}^m$ and suppose we want to turn $f$ into a
immersion. We tackle each partial derivative in turn.
We first make sure $\partial_1f(x) := \partial f(x)/\partial x_1$ is non-zero (for all $x$).
Then we make sure $\partial_2 f(x)$ is not collinear with
$\partial_1f(x)$. Then we make sure $\partial_3f(x)$ is not in the
plane spanned by the two previous derivatives, etc until
all $n$ partial derivatives are everywhere linearly independent.

More generally, if we want $f$ to satisfy some relation (not necessarily
the relation of immersions) then for each $j$ between $1$ and
$n$, we want $\partial_j f(x)$ to live in some open subset
$\Omega_x \subset \mathbb{R}^m$. The brilliant idea is to assume there is a
smooth family of loops
$\gamma \co \mathbb{R}^n \times \SS^1 \to \mathbb{R}^m$ such
that each $\gamma_x$ takes values in $\Omega_x$, and has average value
$\avg{\gamma}_x = \partial_j f(x)$. Obviously such loops
can exist only if $\partial_jf(x)$ is in the convex hull of $\Omega_x$,
hence the name convex integration, and we will see this condition is almost
sufficient. This is where the ampleness condition enters
the story. In the immersion case, this condition holds as soon as $m > n$
because, from the above description, each $\Omega_x$ is the
complement of a subspace with codimension at least two.

For some large positive $N$, we replace $f$ by the new map
\[
  x \mapsto f(x) + \frac1N  \int_0^{Nx_j} \left[\gamma_x(s) - \avg{\gamma}_x\right]ds.
\]
The key is that, provided $N$ is large enough,
we have achieved $\partial_j f(x) \in \Omega_x$, almost without modifying
derivatives $\partial_i f(x)$ for $i \neq j$, and almost without moving
$f(x)$.

This is a very local construction, and it isn't obvious how the absence of
homotopical obstruction, embodied by the existence of a formal solution, should
enter the discussion. The answer is that it essentially provides a way to
coherently choose base points for the $\gamma_x$ loops.

\subsection{Lean and \mathlib}%
\label{sub:lean_and_mathlib}

We formalised all this using the Lean 3 theorem prover \cite{lean},
developed principally by Leonardo de Moura since 2013 at Microsoft
Research. Lean implements a version of the calculus of inductive constructions
\cite{cic}, with quotient types, non-cumulative universes, and proof
irrelevance. Our formalisation is built on top of \mathlib \cite{mathlib-paper},
a library of formalised mathematics with more than 250 contributors and
almost one million lines of code as of September 2022. In return, about two thirds of
the code we wrote for this project has been, or will be,
added to \mathlib. As with many recent projects using \mathlib, the most crucial
property is that library is extremely tightly integrated. We needed to make
simultaneous use of its theories of topology, affine geometry, linear algebra, calculus, and integration.
We also emphasise that \mathlib follows the standard mathematical
practice where nothing is made more complicated than it needs to be for the sake
of avoiding the law of excluded middle or the axiom of choice.

\paragraph{Paper outline: }
\Cref{sec:preliminaries} discusses topics that are pervasive
and not specific to a part of the project.
\Cref{sec:constructing_loop_families} provides the supply of loops.
\Cref{sec:convex_integration} then discusses the proof of the main theorem,
including the key corrugation construction, and \Cref{sec:sphere_eversion}
deduces Smale's sphere eversion theorem.

We will link to specific results in \mathlib and the \textsf{sphere eversion project}
using this icon.\href{https://github.com/leanprover-community/sphere-eversion}{\link}
We will use static links to the version of these repositories when we wrote this paper.

\pagebreak 

\section{Preliminaries}%
\label{sec:preliminaries}

\subsection{Proximity and filters}%
\label{sub:proximity_and_filters}

The kind of geometrical constructions that we consider in this paper frequently
require consideration of properties which hold near a point or near a set.
Informally, this means that there is an unspecified neighbourhood of that point or set where
the property holds, and that this neighbourhood can shrink finitely many times
during the construction. Formally, this idea is perfectly captured by the theory of filters.
This was already
noted in \cite{isabelle_topology} but we will review it here since it
is relevant to the future development of differential topology and since it gives us the
opportunity to advertise a slightly different point of view.

In \cite{isabelle_topology}, filters are mostly seen as providing generalised
bounded quantifiers. We will explain this below, but more generally we see
filters on a type $X$ as generalised sets of $X$. For instance, given a
point $x₀$ in a topological space $X$, we may wish to consider the generalised set
of ``points that are close to $x_0$''. For the types of natural or real
numbers we need the generalised sets of ``very large numbers''. These
generalised sets are called filters on $X$. There is an injective map from the
type $\PS{X}$ of ordinary sets whose elements have type $X$ into the type
$\Fi{X}$ of filters on $X$. One can extend the inclusion order relation to
an order relation $\leq$ on $\Fi{X}$. Given a function $f \co X \to Y$, we can
also extend the associated direct image map from $\PS{X} \to \PS{Y}$ to
$\Fi{X} \to \Fi{Y}$ (and extend also the inverse image map). Using this order relation
and direct image operation, one can for instance say that a sequence $u \co \NN
\to X$ converges to a point $x_0$ if the direct image under $u$ of the filter
of very large natural numbers is contained in the filter of points that are
close to $x_0$.

For any predicate $P$ and any filter $F$ on $X$, one can form the statement
$\forallf{x}{F}, P(x)$ claiming that ``$P(x)$ holds for all $x$ in $F$'', generalising
the statement $∀ x ∈ A, P(x)$ that one can form for any $A ∈ \PS{X}$. However,
contrasting with ordinary sets, the symbols $x ∈ F$ don't mean anything
by themselves. For instance we cannot say that a given real number is very large but we
can say that for all very large real numbers $x$ we have $1/x < 10^{-6}$.
There is also an existential quantifier version which is less crucial but still
nice to keep things symmetric. For instance we can prove that a set $A$ in a
topological space $X$ is open if and only if for every $x$ in $A$, every point
close to $x$ is in $A$. Similarly $A$ is closed if and only if, for every $x$ in $X$,
if there exists a point of $A$ close to $x$ then $x$ is in $A$.

Filters are implemented as sets of sets satisfying three conditions. Thus
$F ∈ \PS{\PS{X}}$ is a filter if:\mllink{order/filter/basic.lean\#L89}
\begin{itemize}
  \item $X \in F$,
  \item $\forall U\, V\in \PS{X}, \; U \in F \text{ and } U ⊂ V \implies V \in F$,
  \item $\forall U\, V\in \PS{X}, \; U \in F \text{ and } V \in F \implies U \cap V \in F$.
\end{itemize}
For instance the filter of points that are close to a point $x_0$ in a topological space
$X$ is the filter $\nhds{x_0}$ of neighbourhoods of $x_0$. The inclusion of
$\PS{X}$ into $\Fi{X}$ sends a set $A$ to the set of sets that contain $A$.
The order relation $\Fi{X}$ is opposite to the order induced by $\PS{\PS{X}}$ (so that
the inclusion $\PS{X} \hookrightarrow \Fi{X}$ is indeed order preserving) and the direct image
operation associated to $f \co X \to Y$ sends $F \in \Fi{X}$ to the set of $U \in \PS{Y}$
such that $f^{-1}(U) \in F$. The statement $\forallf{x}{F}, P(x)$ is implemented as
$\{x | P(x)\} \in F$.

The conditions in the above list translate directly into conditions on the
$\forallf{\dots}{\dots}$ generalised bounded quantifier. Given predicates $P$
and $Q$ on $X$, we get in order:
\begin{itemize}
  \item if $∀ x \in X, P(x)$ then $\forallf{x}{F}, P(x)$,
  \item if $∀ x \in X, P(x) \implies Q(x)$ then\\
    $\big(\forallf{x}{F}, P(x)\big) \implies \big(\forallf{x}{F}, Q(x)\big)$,
  \item if $\forallf{x}{F}, P(x)$ and $\forallf{x}{F}, Q(x)$ then\\
    $\forallf{x}{F}, \big(P(x) \text{ and } Q(x)\big)$.
\end{itemize}

Of course filters seen as generalised sets do not share \emph{all}
properties of ordinary sets. Otherwise one could easily prove the inclusion to
be a bijection and nothing would be gained. For instance the conjunction
property listed above can obviously be extended to conjunctions involving a finite
number of predicates, but not to infinite conjunctions. The
infinite-conjunction extension would imply
that the image of $\PS{X}$ is all of $\Fi{X}$. It would also break the intended
meaning in the case of neighbourhoods.
Furthermore, generalised quantifiers do not commute: given filters $F$ and $G$ on $X$
and $Y$ and given a predicate $P$
on $X \times Y$, the statements $\forallf{x}{F},\, \forallf{y}{G},\; P(x, y)$
and $\forallf{y}{G},\, \forallf{x}{F},\; P(x, y)$ are \emph{not} equivalent in
general.

The moral of this story is that even in formal mathematics one may
write things like ``for all $x$ close to $x_0$, ...'', which might na\"ively sound a bit
sloppy. Better yet, formalising such statements yields a neat theory, together with
precise lemmas saying how to manipulate its statements. In our opinion, this theory
provides an elegant means of capturing our informal intuition.

\subsection{Continuity and differentiability lemmas}%
\label{sub:continuity_lemmas}

We often had to prove that functions were continuous or smooth
in a certain region.
In \mathlib there is a \lstinline{continuity} tactic
that is able to do this for simple functions,
but this tactic can be slow and is often
unable to prove more complicated continuity conditions.
Therefore we usually proved continuity and smoothness conditions manually.

We noticed that the way one formulates continuity and smoothness lemmas
greatly affects the convenience of using them.
For example, consider the statement that addition on $\RR$ is continuous.
We can write this as \lstinline{continuous_add : continuous (λ p : ℝ × ℝ, p.1 + p.2)}.
We would like to be able to use this as follows:
\begin{lstlisting}
example : continuous (λ x : ℝ, x + 3) :=
continuous_add.comp
  (continuous_id.prod_mk continuous_const)
\end{lstlisting}
but Lean rejects this.

To understand why, we remind the reader about Lean's
projection notation (discussed further in ~\cite{Lewis2019Hensel}).
E.g.\\
\lstinline{continuous_id.prod_mk continuous_const} is short for
\lstinline{continuous.prod_mk continuous_id continuous_const}.
We also recall some relevant lemmas from \mathlib:
\begin{lstlisting}
continuous_id : continuous (λ x, x)
continuous_const : continuous (λ x, a)
continuous.comp : continuous g →
  continuous f → continuous (g ∘ f)
continuous.prod_mk : continuous f →
  continuous g → continuous (λ x, (f x, g x))
\end{lstlisting}
The problem is that the elaborator gets confused when unifying the function \lstinline{λ x : ℝ, x + 3}
in the statement with the partially elaborated function
\lstinline{(λ p : ℝ × ℝ, p.1 + p.2) ∘ λ x : ℝ, (x, 3)}
in the proof. We can fix this by giving more information.
The following example is accepted, but is cumbersome:
\begin{lstlisting}
example : continuous (λ x : ℝ, x + 3) :=
continuous_add.comp
  (continuous_id.prod_mk continuous_const :
  continuous (λ x : ℝ, (x, (3 : ℝ))))
\end{lstlisting}
The situation is much nicer if we replace our original continuity statement with:
\begin{lstlisting}
continuous.add : continuous f →
  continuous g → continuous (λ x, f x + g x)
\end{lstlisting}
Using this, the following slick proof succeeds:
\begin{lstlisting}
example : continuous (λ x : ℝ, x + 3) :=
continuous_id.add continuous_const
\end{lstlisting}

In general, instead of stating that a certain function (like addition) is continuous,
we prefer to state that if we apply the function to continuous arguments,
then the resulting function is continuous. This becomes even more important when
the source or target spaces become (products of) spaces of continuous linear maps
and the relevant map is evaluating or composing continuous linear maps since such
cases require a lot more work from the elaborator.
The same holds for many other properties of functions,
like differentiability, smoothness, measurability, and continuity at a point or in a set.
We use this technique throughout our formalisation and mostly also in \mathlib.


We had to prove various lemmas showing that functions defined piecewise were continuous.
A general version of such a lemma (that follows the preceding principle) is the following:\mllink{topology/continuous_on.lean\#L1049}
\begin{lstlisting}
lemma continuous_if {p : α → Prop}
  (hp : ∀ a ∈ frontier {x | p x}, f a = g a)
  (hf : continuous_on f (closure {x | p x}))
  (hg : continuous_on g (closure {x | ¬p x})) :
  continuous (λ a, if p a then f a else g a)
\end{lstlisting}
We will use a special case of this lemma
where \lstinline {p x} is an inequality between two continuous functions%
\mllink{topology/algebra/order/basic.lean\#L532}
in \Cref{lem:satisfied_or_refund}.

\subsection{Convexity}
\label{sub:convexity}
Affine spaces over ordered scalars have a much richer structure than those over
general scalars. An especially important concept that arises in the presence of
a scalar ordering is that of convexity.

For our application, we were concerned with convexity of subsets of vector
spaces over the real numbers. Given a real vector space $F$ together with a
subset $S \subseteq F$, one says $S$ is convex if for all $x$, $y$ in $S$, the
line segment joining $x$ and $y$ is contained in $S$.

\mathlib contains a substantial library of results about convexity which we
found especially useful. There are various ways of defining convexity. In fact
\mathlib's definition\mllink{analysis/convex/basic.lean\#L43} changed several
times during the course of our work and on several occasions we had to make
adjustments in our work to cater for this. Such adjustments became easier as
the convexity API stabilised and we regard this as encouraging empirical
evidence in favour of using \mathlib as a dependency.

While \mathlib already contained key concepts such as the convex hull%
\mllink{analysis/convex/hull.lean\#L35} of a set,
of course not every result we needed about convexity theory already existed.
Missing results were added directly to \mathlib. Examples include
Carathéodory's lemma,\mllink{analysis/convex/caratheodory.lean\#L143} a characterisation of
when a convex set has non-empty
interior,\mllink{analysis/normed_space/add_torsor_bases.lean\#L159} and numerous
minor technical lemmas. These are all now available to any future \mathlib
consumer.

\subsection{Barycentric coordinates and their smoothness}
\label{sub:barycentric}
An affine basis for a $d$-dimensional real vector space $F$ is a set of $d+1$
points $p_0, p_1, \ldots, p_d$ in $F$ such that any $q$ in $F$ can be written
as $q = \sum w_i p_i$ for unique scalars $w_i$ such that $\sum w_i = 1$.
The scalars $w_i$ are known as the barycentric coordinates of $q$ with respect
to the basis $p_i$.

Barycentric coordinates are the natural coordinates when working with
affine-invariant concepts such as convexity. Several of our key constructions
required the use of barycentric coordinates and so we added a
definition\mllink{linear_algebra/affine_space/basis.lean\#L92} and
corresponding API to \mathlib. Amongst the basic facts we needed to add were:
\begin{itemize}
  \item any real vector space\footnote{Actually we proved this for affine space
  over any division ring.} has an affine
  basis,\mllink{linear_algebra/affine_space/basis.lean\#L367}
  \item a point lies in the convex hull of an affine basis
  iff all of its barycentric
  coordinates are  non-negative,\mllink{analysis/convex/combination.lean\#L441}
  \item a point lies in the \emph{interior} of the convex hull of an affine
  basis iff all of its barycentric coordinates are
  strictly positive.\mllink{analysis/normed_space/add_torsor_bases.lean\#L57}
\end{itemize}

A slightly trickier result we needed concerned the smoothness of barycentric coordinates. Let
$A \subseteq F^{d+1}$ be the set of affine bases of $F$. The result we needed
was that the map $F \times A \to \R^{d+1}$, $(q, p_0, \ldots, p_d) \mapsto (w_0, \ldots, w_d)$
is smooth.\selink{to_mathlib/smooth_barycentric.lean\#L119}
Note that this is joint smoothness on $F \times A$ so both the point in $F$ and
the basis in $A$ are varying simultaneously. This is an example of a
geometrically obvious fact that wouldn't receive any explanation in an informal
context.

Although one could prove this result using the language of the exterior algebra,
we opted for a more elementary proof using determinants. At the time \mathlib's theory of
determinants was more complete than its theory of the exterior algebra (see also \cite{Wieser2022}).
The key ingredients were:
\begin{itemize}
  \item a barycentric change-of-basis formula,\mllink{linear_algebra/affine_space/basis.lean\#L288}
  \item the smoothness of the determinant
  function.\selink{to_mathlib/smooth_barycentric.lean\#L83}
\end{itemize}

\subsection{Partitions of unity}%
\label{sub:partitions_of_unity}

Partitions of unity are an important tool in differential topology. They often allow one
to patch a family of local solutions to some problem into the desired global solution. They are used
so often that it is challenging to render a complete informal presentation without becoming
repetitive. This kind of
repetition is especially bad from a formalisation point of view so we extracted an ``induction principle''
associated to partitions of unity.

\mathlib already contained a sizeable library of results about partitions of unity, including versions for
normal topological spaces and versions for smooth manifolds. Notably,
there is no version specialised to normed vector spaces and so one must regard the vector
spaces as manifolds in order to invoke \mathlib's API for partitions of unity.
Technically this comes at the cost of making partitions of unity for vector spaces
depend on a certain amount of \mathlib's differential geometry library. However
from the user's point of view, this is completely hidden as can be seen in the
lemmas below.

In the first lemma we want to construct maps $f \co E \to F$ such that $\forall x, P(x, f(x))$
for some predicate $P$ on $E \times F$. The assumptions are a convexity assumption on $P$
and the existence of local solutions. Partitions of unity appear in the proof but not in
the statement.\selink{to_mathlib/partition.lean\#L283}

\begin{lemma}
  \label{lem:exists_cont_diff_of_convex}
  Let $E$ and $F$ be real normed vector spaces. Assume that $E$ is finite
  dimensional. Let $P$ be a predicate on $E \times F$ such that for all $x$ in
  $E$, $\{y ~|~ P (x, y) \}$ is convex. Let $n$ be a natural number or $+\infty$.
  Assume that every $x$ has a neighbourhood $U$ on which there exists a $\C^n$
  function $f$ such that $\forall x ∈ U, P(x, f(x))$. Then there is a global
  $\C^n$ function $f$ such that $\forall x, P(x, f(x))$.
\end{lemma}

As a simple application of this lemma, one can prove that, given any function
$f \co E \to \RR$, if $f$ is locally bounded below by a positive constant then
$f$ is globally bounded below by a smooth positive function. Indeed the convexity
assumption is satisfied because for each $x$ the interval $(0, f(x)]$ is convex
and the local boundedness assumption provides (constant) local solutions.

More generally, this lemma allows one to deduce useful approximation results
that would be painful to obtain using obvious alternatives such as Stone-Weirstrass or convolution.
In particular, we will use the following.\selink{to_mathlib/partition.lean\#L328}
\begin{corollary}
  \label{cor:exist_smooth_approximation}
  Given $n$, $E$ and $F$ as in the lemma, let $C$ be a closed set in $E$,
  let $\e : E \to \RR$ be a continuous positive function and let $f \co E \to F$ be
  a continuous function. If $f$ is of class $\C^n$ near $C$ then there exists a function
  $f' \co E \to F$ which is everywhere of class $\C^n$, coincides with $f$ on $C$ and
  such that $\|f'(x) - f(x)\| \le \e(x)$ everywhere.
\end{corollary}

The abstraction captured in \Cref{lem:exists_cont_diff_of_convex} also exposes
a mathematical awkwardness that exists for smooth maps, both formally and informally.
The problem is that smooth manifolds and maps do not form a
cartesian closed category; equivalently, one cannot curry smooth functions.

The point is that, even for finite-dimensional vector spaces $E$, $F$ (where there are
no topological issues) there is no natural smooth
structure on the space of smooth maps $E \to F$. For our work, we needed a variant
of \Cref{lem:exists_cont_diff_of_convex} for a product space $E = E_1 \times E_2$.
One would like to be able to deduce this from \Cref{lem:exists_cont_diff_of_convex}
by taking $E$ and $F$ to be $E_1$ and $E_2 \to F$ respectively but there is no
norm on $E_2 \to F$ and so one cannot talk about smooth maps taking values there.

Even the informal mathematics community has not converged on a solution to this
problem\footnote{For example though the category of Fr\"olicher spaces is Cartesian
closed, it is hardly used and is not locally Cartesian closed.}. We thus used
partitions of unity to prove a second lemma, similar to
\Cref{lem:exists_cont_diff_of_convex} except for products
$E_1 \times E_2$.\selink{to_mathlib/partition2.lean\#L121} It is notable that
formalism forced us to confront this usually harmless fact about the category
of smooth manifolds.

\section{Constructing loop families}%
\label{sec:constructing_loop_families}

\subsection{Loop families}%
\label{sub:loop_families}

The corrugation procedure outlined in \Cref{sec:introduction} requires having
families of loops taking values in some prescribed set, with some prescribed
base point, and with prescribed average value. All those loops should also come
with a global homotopy from the family of constant loops at the prescribed base
points. Everything is gathered in the following proposition
where $\Omega$ prescribes values, the map $\beta$ prescribes base points, and $g$
prescribes the average values. The parameter $t \in \RR$ is the homotopy parameter
while $s \in \SS^1$ is the loop parameter. We fix finite-dimensional
normed vector spaces $E$ and $F$ over $\RR$.

\begin{proposition}
  \label{prop:loops}
  Let $K$ a compact set in $E$. Let $\Omega$ be an open set in $E \times F$.

  Let $\beta$ and $g$ be smooth maps from $E$ to $F$.
  Write $\Omega_x := \{ y \in F \mid (x, y) \in \Omega\}$, assume that
  $\beta(x) \in \Omega_x$ for all $x$, and that $g(x) = \beta(x)$ for $x$ near
  $K$.
  Lastly suppose that for every $x$, $g(x)$ is in $\CH(\conn{\Omega_x}{\beta(x)})$,
  the convex hull of the connected component of $\Omega_x$ containing $\beta(x)$,
  then there exists a smooth family of loops
  \[
    \gamma \co E \times \RR \times \SS^1 \to F, (x, t, s) \mapsto \gamma^t_x(s)
  \]
  such that, for all $x \in E$, all $t \in \mathbb{R}$ and all  $s \in \SS^1$,
  \begin{itemize}
    \item
      $\gamma^t_x(s) \in \Omega_x$,
    \item
      $\gamma^t_x(s) = \beta(x)$ when $t = 0$ or $s = 0$ or $x$ is near $K$,
    \item
      the average of $\gamma^1_x$ is $g(x)$.
  \end{itemize}
\end{proposition}

Part of the challenge in formalising this lemma\selink{loops/exists.lean\#L253}
is that there is a lot of freedom in constructing $\gamma$,
which is problematic when trying to do it consistently when $x$ varies. This issue
will be addressed in \Cref{sub:surrounding_families}.

From a formalisation point of view, two questions must be addressed from the
beginning. The first one is how to handle maps defined on $\S$ and the second
one is juggling between curryfied and uncurryfied functions.

On paper the notation $\S$ is already ambiguous. It could refer to the unit
circle in $\RR^2$ (or in $\CC$) or to a quotient of $\RR$, by the subgroup
$\ZZ$ or $2\pi\ZZ$ depending on taste. Note that we need to be able to talk
about continuous and smooth maps defined on $\S$. We ultimately settled on the
following:\selink{loops/basic.lean\#L36}
\begin{lstlisting}
structure loop (X : Type*) :=
(to_fun : ℝ → X)
(per' : ∀ t, to_fun (t + 1) = to_fun t)
\end{lstlisting}
which packages a function \lstinline{to_fun} with a periodicity assumption
(with period $1$ so that our implementation is closest to $\RR/\ZZ$). We then record a coercion
from \lstinline{loop X} to \lstinline{ℝ → X} so that, given \lstinline{t : ℝ}
and \lstinline{γ : loop X} one can write \lstinline{γ t}.
Note in particular that there is no type $\S$ in this
story, only a type which plays the role of functions on $\S$ but is not a
function type from a foundational point of view. One last note about \lstinline{loop}:
the reason why the \lstinline{per'} field of \lstinline{loop} is named with
a prime is that it is later restated as \lstinline{loop.per} in terms of this
coercion to function.\selink{loops/basic.lean\#L79}

With this definition of loops, it is very easy to state that a loop is
continuous or smooth. But this is not enough. We also need families of loops,
parametrised by topological spaces. In particular we also need loops parametrised by a
normed space $E$, or by $E \times \RR$ as in the statement of \Cref{prop:loops}.
This creates some tension since we would like to think of such a family of loops
as a function on $E \times \R \times \S$ but our loops are not true function types
and so we must do some extra work in order to obtain the usual conveniences of
partial evaluation and currying when working with families of loops.
We thus introduced a type class:\mllink{logic/function/basic.lean\#L643}
\begin{lstlisting}
class has_uncurry (α β γ : Type*) :=
(uncurry : α → β → γ)
\end{lstlisting}
which records a way to turn an element of \lstinline{α} into a function from
\lstinline{β} to \lstinline{γ}, with notation \lstinline{↿} for \lstinline{has_uncurry.uncurry}.
The most generic use is to uncurry recursively. For instance a function
\lstinline{f : α → β → γ → δ} will be fully uncurried to \lstinline{↿f : α × β × γ → δ}.
Using this typeclass, for every pair of types \lstinline{α} and \lstinline{X}, we then register
the key instance \lstinline{has_uncurry (α → loop X) (α × ℝ) X}
which allows us to convert any function \lstinline{φ : α → loop X} to a function \lstinline{α × ℝ → X}.
This way we can state the smoothness conclusion from \Cref{prop:loops} simply as
\lstinline{𝒞 ∞ ↿γ}.

This setup is not completely bullet-proof: sometimes the elaborator gets confused
and needs some help, despite the fact that, contrary to the slightly simplified code
displayed above, the actual code declares \lstinline{β} and \lstinline{γ} as output
parameters for type class instance search. However we are globally satisfied by this
encoding.

\subsection{Surrounding families}
\label{sub:surrounding_families}

We now discuss the proof of a version of \Cref{prop:loops},
subject to two simplifications.
Firstly, we work with \emph{continuous} families of loops.
We will smooth these families
at the end, taking advantage of the fact that $\Omega$ and the surrounding
condition are open.
Secondly, we work with families of loops that don't have a prescribed average, but
which we can reparametrise to have the prescribed average.
We will do the reparametrisation in \Cref{sub:reparametrisation}.
To ensure that this reparametrisation exists,
we need to require that the loop $\gamma_x$ \emph{surrounds} $g(x)$.

\begin{definition}
A loop $\gamma$ \emph{surrounds} a point $v$ if there is an affine basis in the image of $g$
such that $v$ has positive barycentric coordinates with respect to this basis.\selink{loops/surrounding.lean\#L344}
\end{definition}
From the discussion in \Cref{sub:barycentric} it seems that
we could have given the definition equivalently as $v \in \CH(\im\gamma)^o$.
This is indeed the definition used in the standard references \cite{Gromov_PDR, Spring}.
However, this is not clearly equivalent.
Notice that $v \in \CH(A)^o$ does not always imply that
there is an affine basis in $A$ such that $v$ has positive barycentric coordinates
with respect to this basis.
As a counterexample, consider $A$ to be the vertices of a square in the plane,
and let $v$ be the center of the square.
We define a loop surrounding a point as above, because this is exactly the condition we need.

The first main task in proving the special case of \Cref{prop:loops} is to construct
suitable families of loops $\gamma_x$ surrounding $g(x)$, by assembling local
families of loops.
We therefore introduce the following definition.\selink{loops/surrounding.lean\#L523}

\begin{definition}
  \label{def:family_surrounds}
  A continuous family of loops
  $\gamma : E \times [0, 1] \times \SS^1 \to F, (x, t, s) \mapsto \gamma^t_x(s)$
  surrounds a map $g : E \to F$ with base $\beta : E \to F$ on $U \subseteq E$ in
  $\Omega \subseteq E \times F$ if, for every $x$ in $U$, every $t \in [0, 1]$ and every
  $s \in \SS^1$,
  \begin{itemize}
    \item
      $\gamma^t_x(s) = \beta(x)$ if $t = 0$ or $s = 0$,
    \item
      $\gamma^1_x$ surrounds $g(x)$,
    \item
      $(x,\gamma^t_x(s)) \in \Omega$.
  \end{itemize}
  The space of such families will be denoted by
  $\Loop(g, \beta, U, \Omega)$.
\end{definition}

In this section we assume the hypotheses of \Cref{prop:loops},
i.e. $\beta$ and $g$ are smooth maps, $\beta(x)\in\Omega_x$ for all $x$,
and $g(x) \in \CH(\conn{\Omega_x}{\beta(x)})$.

Using Carathéodory's lemma, we can construct a surrounding loop at a single point $x$ and thus
obtain an element of $\Loop(g, \beta, \{x\}, \Omega)$.\selink{loops/surrounding.lean\#L491}
Since $g(x) \in \CH(\conn{\Omega_x}{\beta(x)})$ and $\conn{\Omega_x}{\beta(x)}$ is open,
we can obtain an affine basis $B\subseteq\conn{\Omega_x}{\beta(x)}$ such that
$g(x) \in \CH(B)^o$.\selink{loops/surrounding.lean\#L140}
Since $\conn{\Omega_x}{\beta(x)}$ is path-connected,
we can then find a path in $\conn{\Omega_x}{\beta(x)}$ starting at $\beta(x)$ through all points in $B$.
To make it a null-homotopic loop in $\Omega_x$ based at $\beta(x)$,
we traverse the same path backward.
This homotopy provides an element of $\Loop(g, \beta, \{x\}, \Omega)$.

Moreover, we can even construct families of surrounding loops locally around a point $x_0$.%
\selink{loops/surrounding.lean\#L751}
We take our element $\gamma\in\Loop(g, \beta, \{x_0\}, \Omega)$ and set
\[\gamma_x^t(s)=\gamma^t(s)+\beta(x)-\beta(x_0).\]
Since $\Omega$ is open and barycentric coordinates are smooth,
this will give a surrounding family in $\Loop(g, \beta, U, \Omega)$
for some neighbourhood $U$ of $x_0$.

The difficulty in constructing global families of surrounding loops is that
there are plenty of surrounding loops and we need to choose them
consistently.
The key feature of the above definition is that the parameter $t$ not only
allows us to carry out the corrugation
process in the next section, but also brings a ``satisfied or refund'' guarantee,
as explained in the next lemma.\selink{loops/surrounding.lean\#L964}

\begin{lemma}
  \label{lem:satisfied_or_refund}
  For any set $U \subseteq E$,  $\Loop(g, \beta, U, \Omega)$ is contractible:
  for every $\gamma_0$ and $\gamma_1$ in $\Loop(g, \beta, U, \Omega)$,
  there is a continuous map
  $\delta : [0, 1] \times E \times [0, 1] \times \SS^1 \to F,
  (\tau, x, t, s) \mapsto \delta^t_{\tau, x}(s)$
  which interpolates between $\gamma_0$ and $\gamma_1$ in
  $\Loop(g, \beta, U, \Omega)$.
\end{lemma}

The tricky part of this lemma is that we need to make sure that $\delta$ always
surrounds $g$. The informal proof is again a nice picture: the idea is to
build a path of loops that starts with $\gamma_0$ then $\gamma_0$
concatenated with a longer and longer initial segment of $\gamma_1$ until one reaches
the full concatenation of $\gamma_0$ and $\gamma_1$ at $\tau = 1/2$, and start
replacing $\gamma_0$ by a shorter and shorter initial segment. For each $\tau$
this contains a full copy of either $\gamma_0$ or $\gamma_1$ hence surrounds $g$.
We now describe how we implemented this picture.

Let $\rho:\RR\to\RR$ be a piecewise-affine function with
$\rho(t)=1$ for $t\le\frac12$ and $\rho(t)=0$ for $t\ge1$.
We can define the homotopy $\delta$ as follows:
\begin{itemize}
  \item $\delta_{\tau,x}^t$ moves along the loop
    $\gamma_{0,x}^{\rho(\tau)t}$ on $[0,1-\tau]$ (if $\tau<1$)
  \item $\delta_{\tau,x}^t$ moves along the loop
    $\gamma_{1,x}^{\rho(1-\tau)t}$ on $[1-\tau,1]$ (if $\tau>0$)
\end{itemize}
Note that the image of $\delta_{\tau,x}^1$ contains
the image of $\gamma_{0,x}^1$ for $\tau\le\frac12$,
and the image of $\gamma_{1,x}^1$ for $\tau\ge\frac12$.
Hence it always surrounds $g(x)$.

The argument that $\delta$ is continuous is surprisingly tricky.
First of all, $\delta$ is defined piecewise,
so we have to check that the different cases agree on the frontier.
Furthermore, note that if $\tau\to1$ then $\delta_{\tau,x}^1$
will move along loop $\gamma_{0,x}^{\rho(\tau)t}$
at a speed that tends to $+\infty$,
so we need to show that $\gamma_{0,x'}^{\rho(\tau)t'}$
tends uniformly to the constant loop as $(x',\tau,t')\to(x,1,t)$,
which follows from the fact that $\gamma$ is continuous.

Using this lemma, we can transition between two solutions.
Therefore, if we have a solution $\gamma_i$ near $K_i$ for a compact sets $K_i$ ($i\in\{0,1\}$),
we can find a solution near $K_0\cup K_1$ that coincides with $\gamma_0$ near $K_0$.%
\selink{loops/surrounding.lean\#L995}

Finally, we can apply this recursively to obtain the following result.%
\selink{loops/surrounding.lean\#L1257}
\begin{lemma}
  \label{lem:exists_surrounding_loops}
  In the setup of \Cref{prop:loops}, assume we have a
  continuous family $\gamma$ of loops defined near $K$ which is based at $\beta$,
  surrounds $g$ and such that each $\gamma_x^t$ takes values in $\Omega_x$.
  Then there such a family which is defined on all of $E$ and agrees
  with $\gamma$ near $K$.
\end{lemma}

The proof requires finding a countable locally-finite family of compact sets covering $E$
and extending the solution recursively.

\subsection{The reparametrisation lemma}
\label{sub:reparametrisation}
The reparametrisation lemma concerns the behaviour of the average value of a
smooth loop $\gamma : \S → F$ when the loop is reparametrised by precomposing
it with a diffeomorphism $\phi : \S → \S$.

Given a loop $\gamma : \S → F$, for some finite-dimensional real vector space
$F$, one may integrate to obtain its average $\avg{\gamma} = \int_0^1 \gamma$.
Although this average depends on the loop's
parametrisation\footnote{Intuitively, the parametrisation is the speed at which
$\gamma(s)$ moves when $s$ moves with unit speed in $\SS^1$.}, it satisfies a
constraint that depends only on the image of the loop: the average is contained
in the closure of the convex hull of the image of $\gamma$.
Indeed the integral defining the average value is a limit of average values over a
finite sample of values and those finite averages belong to $\CH(\im \gamma)$.

The reparametrisation lemma says that conversely, given any point $g$ \emph{surrounded}
by $\gamma$, there exists a reparametrisation $\phi$ such that $\gamma \circ \phi$ has
average value $g$.

The reparametrisation lemma thus allows one to reduce the problem of
constructing a loop whose average is a given point, to the problem of
constructing a loop subject to a condition that depends only on its image.

The idea of the proof is simple: since $g$ is contained in the
the convex hull of the image of $\gamma$, there exist $s_0, s_1, \ldots, s_d$
and barycentric coordinates $w_0, w_1, \ldots, w_d$ such that:
\begin{align*}
  g = \sum w_i \gamma (s_i).
\end{align*}
If there were no smoothness requirement on $\phi$ one could define it to be a
step function which spends time $w_i$ at each $s_i$. However because there is a
smoothness condition, one must approximate by rounding off the corners of the
would-be step function. Using such an
approximation it is easy to see that the average of $\gamma \circ \phi$ can be
made arbitrarily close to $g$. In order to find $\phi$ such that the average is
exactly $g$ we use the additional freedom that we may also vary the $w_i$.
Because being surrounded is an open condition, a simple continuity
argument shows that this additional freedom is sufficient.

Because the $s_i$ are constant, it is easy to construct the inverse of $\phi$,
which is what we did. It is constructed as the integral of a sum of
approximations to the Dirac delta functions, which we call delta
mollifiers.\selink{loops/delta_mollifier.lean\#L270}

In fact the reparametrisation lemma holds for \emph{families} of loops and this
was the version that we needed. More precisely we proved the following:
\begin{lemma}
  Let $E$, $F$ be a finite-dimensional normed real vector
  spaces, $\gamma$ a smooth family of loops:
  \begin{align*}
    \gamma : E \times \S &\to F,\\
    (x, s) &↦ \gamma_x (s),
  \end{align*}
  and $g : E → F$ a smooth function such that $\gamma_x$ surrounds $g(x)$ for
  all $x$.
  Then there exists a smooth family of diffeomorphisms $\phi_x$ of $\SS^1$
  such that $\avg{\gamma_x ∘ \phi_x} = g(x)$ and $\phi_x(0) = 0$ for all $x$.
\end{lemma}
The argument outlined above for a single loop works locally in the
neighbourhood of any $x$ in $E$ and one uses a partition of unity to globalise
all the local solutions into the required family.

Actually in our formalisation, the statement of this lemma is distributed across a
definition,\selink{loops/reparametrization.lean\#L524} a lemma about its
smoothness,\selink{loops/reparametrization.lean\#L553} and a lemma about its
average values.\selink{loops/reparametrization.lean\#L564}

\subsection{Proof of the loop construction proposition}
\label{sub:proof_loop_prop}
Using these ingredients, we can now prove \Cref{prop:loops}.
\begin{proof}[Proof sketch of \Cref{prop:loops}]
Let $\gamma^*$ be any family of loops in $\Loop(0,0,\{0\},F)$.
In a neighborhood $U^*$ of $K$ where $g=\beta$, we can set $\gamma'_x=g(x)+\epsilon\gamma^*$.
Here $\epsilon>0$ must be small enough to ensure that $\gamma'_x$ lands in $\Omega_x$,
which is possible since $\Omega$ is open and $K$ is compact.
From \Cref{lem:exists_surrounding_loops} we obtain
a continuous family of surrounding loops $\gamma'_x$ for all $x$.
We can now approximate $\gamma'$ with a smooth family of loops $\gamma^S$.
Next, we reparametrise $\gamma^S_x$ as discussed in \Cref{sub:reparametrisation} to obtain a
smooth family of loops $\gamma^R_x$ with average $g(x)$.
Finally, we use a smooth cut-off function to transition between $g(x)$ near $K$
to $\gamma^R_x$ outside $U^*$ to obtain our final family $\gamma$ that equals $g$ near $K$.
\end{proof}
There are a couple of nuances to this argument.
First, we have to ensure that all our constructions remain in $\Omega$.
To do this, we must strengthen the condition on $U^*$.
We can require that there is a $\delta>0$ such that
for all $x\in U^*$ the ball with center $(x,\beta(x))$ and radius $2\delta$ lies on $\Omega$
and that the distance between $\gamma'_x$ and $\beta(x)$ is at most $\delta$.
When smoothing, we then require that $\gamma^S$ lies at most $\delta$ from $\gamma'$.
This ensures that $\gamma$ (which lies on the segment from $g$ to $\gamma^R$) lies in $\Omega$.

A second nuance is that we need to make sure that the smoothed family $\gamma^S$
still surrounds $g(x)$.
We ensure this by requiring that $\gamma^S$ is close enough to $\gamma'$ and invoking
a lemma that states that all loops close enough to a given loop
still surround a given point.\selink{loops/surrounding.lean\#L610}

A third nuance is the question of how we obtain a smooth function near a continuous one.
Our first plan was to use a convolution with a smooth bump function.
We need to require that $(\gamma^S)_x^t(s)=\beta(x)$ on $C = \{(t, s) \mid t = 0 \vee s = 0\}$.
We planned to continuously reparametrise $\gamma'_x$
so that it becomes constant near $C$ and then
use the fact that the convolution of a function that is constant near $x_0$
with a bump function with small enough support doesn't change the value at $x_0$.
However, the problem is that we need to smooth $\gamma'$ in all arguments $(x, t, s)$, and
$(\gamma')_x^t(s)$ varies as $x$ varies, even near $C$,
since in that region it equals $\beta(x)$.

We did not find a way to solve this problem with convolutions,
since convolutions do not give you enough control over the resulting function in this case.
Instead, we used an argument based on partitions of unity.
We use the same argument to ensure that $(\gamma')_x^t(s)=\beta(x)$ near $C$ (which is smooth!)
and then we apply \Cref{cor:exist_smooth_approximation}.


After taking these nuances into account,
we finally obtain a proof of \Cref{prop:loops}.\selink{loops/exists.lean\#L253}

\section{Convex integration}%
\label{sec:convex_integration}

\subsection{A theorem giving parametricity for free}%
\label{sub:a_theorem_giving_parametricity_for_free}

In this section we explain how to reduce \Cref{thm:gromov} to the case
where the parameter space $P$ is trivial.

Denote by $\Psi$ the map from $J^1(E \times P, F)$ to $J^1(E, F)$ sending $(x,
p, y, \psi)$ to $(x, y, \psi \circ \iota_{x, P})$ where
$\iota_{x, P} : E \to E \times P$ sends $v$ to $(v, 0)$.

To any family of sections $F_p : x \mapsto (f_p(x), \varphi_{p, x})$ of $J^1(E, F)$, we
associate the section $\bar F$ of $J^1(E \times P, F)$ sending $(x, p)$ to
$\bar F(x, p) := (f_p(x), \varphi_{p, x} \oplus \partial f/\partial p(x, p))$.%
\selink{local/parametric_h_principle.lean\#L128}

\begin{lemma}
  \label{lem:parametricity}
  In the above setup, we have:
  \begin{itemize}
    \item
      $\bar F$ is holonomic at $(x, p)$ if and only if $F_p$ is holonomic
      at $x$.\selink{local/parametric_h_principle.lean\#L169}
    \item
      $F$ is a family of formal solutions of some $\Rel \subset J^1(E, F)$ if and
      only if $\bar F$ is a formal solution of $\Rel^P := \Psi^{-1}(\Rel)$.%
      \selink{local/parametric_h_principle.lean\#L154}
    \item
        If $\Rel$ is ample then, for any parameter space $P$, $\Rel^P$ is also ample.%
        \selink{local/parametric_h_principle.lean\#L110}
  \end{itemize}
\end{lemma}

As far as we know, the last item is new. Indeed in an informal account it does not
cost much to write that handling parameters only requires complicating notation
or proving variations of known lemmas, so the incentive to prove the above
lemma is low.
Using it, we obtain the parametricity in \Cref{thm:gromov} for free.%
\selink{local/parametric_h_principle.lean\#L296} 
\begin{lemma}
  \label{lem:parametricity_for_free}
  If $\Rel^P$ satisfies the $h$-principle
  (i.e. the conclusion of \Cref{thm:gromov} for all appropriate $C$, $K$ and $\F_0$)
  with a trivial parameter space $P$,
  then $\Rel$ satisfies the $h$-principle with parameter space $P$.
\end{lemma}

\subsection{Corrugations}%
\label{sub:corrugations}

In this section we comment on our formalisation of Theillière's corrugation
operation introduced in \cite{Theilliere22}. In fact for our purposes we need to
generalise the results of \cite{Theilliere22} slightly.

Fix a dual pair $p = (\pi, v)$ on $E$. Recall that this means
$\pi \in E^*$, $v \in E$, and $\pi(v) = 1$. Given a family of loops $\gamma_x$ in $F$ parametrised
by $x$ in $E$, and a real number $N$ the corrugation
map $\Corr{p}{N}{\gamma} \co E \to F$ is defined by:\selink{local/corrugation.lean\#L56}
\[
\Corr{p}{N}{\gamma}(x) = \frac1N \int_0^{N\pi(x)} \left(\gamma_x(s) - \avg{\gamma}_x\right)\, ds.
\]
We also define the remainder term:
$\Rem{p}{N}{\gamma} := \Corr{p}{N}{\partial_x\gamma}$ where
$\partial_x\gamma \co E \times \S \to \L{E}{F}$ is the partial derivative of $\gamma$
in the direction of $E$.

\begin{proposition}\label{prop:corrugation}
  Let $\gamma \co [0, 1] \times E \times \S \to F$ be a smooth family of loops in $F$
  parametrised by $[0, 1] \times E$.
  Let $K$ be a compact set in $E$ and let $\e$ be a positive real number. Then:
  \begin{itemize}
    \item
      $(t, x) \mapsto \Corr{p}{N}{\gamma_t}(x)$ is smooth,\selink{local/corrugation.lean\#L138}
    \item for every large $N$ and every $x$ in $K$,
      $\|\Corr{p}{N}{\gamma_t}(x)\| \le \e$,\selink{local/corrugation.lean\#L105}
    \item for every large $N$ and every $x$ in $K$,
      $\|\Rem{p}{N}{\gamma_t}(x)\| \le \e$,\selink{local/corrugation.lean\#L237}
    \item
      for every $t$,
      \[
        D\Corr{p}{N}{\gamma_t}(x) = \pi \otimes \left(\gamma_{t, x}(N\pi(x)) - \avg{\gamma}_{t, x}\right) +
      \Rem{p}{N}{\gamma_t}(x).\selink{local/corrugation.lean\#L186}
      \]
  \end{itemize}
\end{proposition}

The first point in the above proposition wouldn't be stated in an informal context, let alone
be proven. If pressed to provide a hint of proof, we would say this map is smooth
as a composition of smooth maps. We already discussed how to state and prove
smoothness lemma in \Cref{sub:continuity_lemmas}. The next ingredient of
course is a strong library of calculus and integration. At the beginning of this
project, \mathlib already contained such libraries,
including the fundamental theorem of calculus (see \cite{vandoorn2021haar} for
general explanation about integration and measure theory in \mathlib).
However it had nothing about parametric integrals. We thus needed to add the following lemma.%
\selink{to_mathlib/measure_theory/parametric_interval_integral.lean\#L539}%
\selink{to_mathlib/measure_theory/parametric_interval_integral.lean\#L404}

\begin{lemma}
  \label{lem:parametric_integral}
  Let $F$ be a real Banach space and $H$ be a finite-dimensional real normed space.
  Let $n$ be a natural number or $+\infty$, and let $a$ be a real number. Assume
  $s \co H \to \RR$ and $\Phi \co H \times \RR \to F$ are of class $\C^n$.
  Then the function defined by $x \mapsto \int_a^{s(x)} \Phi(x, t)\, dt$
  is of class $\C^n$ and, assuming $n > 0$, its derivative is
  \[
    x \mapsto \int_a^{s(x)} \frac{\partial \Phi}{\partial x}(x, t)\, dt + \Phi(x, s(x)) \otimes ds(x).
  \]
\end{lemma}

Note this lemma includes a version of the fundamental theorem of calculus
when $H = \R$, $\Phi$ does not depend on $x$, and $s = \id$. So this theorem is
obviously an ingredient in the above lemma. The other ingredient is the
dominated convergence theorem which allows one to swap the order integration and
differentiation when $s$ is constant. One might argue that invoking dominated convergence is excessive in our
situation where $F$ is finite dimensional, we assume continuous
differentiability everywhere, and the integration domain is compact. However we
proved these lemmas in the broader context of building \mathlib which strives
to be a general-purpose mathematics library. We thus first prove much
more general lemmas. The following is a sample statement.%
\selink{to_mathlib/measure_theory/parametric_interval_integral.lean\#L264}

\begin{lemma} 
  \label{lem:parametric_integral_bis}
  Let $F$ be a real Banach space and $H$ be a real normed space.
  Let $a_0$, $b_0$, $a$ and $\e$ be real numbers such that $a \in (a_0, b_0)$ and
  $\e > 0$. Let $x_0$ be a point in $H$.
  Let $s \co H \to \RR$, $\Phi \co H \times \RR \to F$, $\Phi' \co \RR \to \L{H}{F}$,
  and $b \co \RR \to \RR$ be functions.
  Suppose that the following properties hold:
  \begin{itemize}
    \item $s(x_0) \in (a_0, b_0)$ and $s$ is differentiable at $x_0$ with derivative
      $s' \in \L{H}{\R}$;
    \item $\Phi(x, \cdot)$ is almost everywhere strongly measurable on $(a_0, b_0)$
      for all $x \in B(x_0, \e)$;
    \item $\Phi(x_0, \cdot)$ is integrable on $(a_0, b_0)$ and continuous at $s(x_0)$;
    \item $\Phi'$ is almost everywhere strongly measurable on $(a, s(x_0))$;
    \item for almost every $t$ in $(a_0, b_0)$, $\Phi(\cdot, t)$ is
      $b(t)$-Lipschitz on the ball $B(x_0, \e)$;
    \item $b$ is integrable on $(a_0, b_0)$, non-negative and continuous at $x_0$;
    \item for almost every $t$ in $(a, s(x_0))$, $\Phi(\cdot, t)$ is differentiable at $x_0$
      with derivative $\Phi'(t)$.
  \end{itemize}
  Then $\Phi'$ is integrable on $(a, s(x_0))$ and
  the function defined by $x \mapsto \int_a^{s(x))} \Phi(x, t)\, dt$
  is differentiable at $x_0$ with derivative
  \[
    \int_a^{s(x_0)} \Phi'(t)\, dt + \Phi(x_0, s(x_0)) \otimes s'.
  \]
\end{lemma}

We have not been able to find the above statement in any textbook. It has
rather minimalistic assumptions that are quite subtle. For instance the positive radius $\e$
is meant to ensure some uniformity in $t$ when requiring that $\Phi(\cdot, t)$ is
$b(t)$-Lipschitz near $x_0$. Also note that differentiability of $\Phi(\cdot, t)$ is assumed
only at $x_0$ and we do not require any bound on $\Phi'$, this is deduced from the local
Lipschitz assumption. We have many variations on this lemma including versions
where the measure used isn't the Lebesgue measure. They would be rather difficult to write
without a proof assistant (we certainly wrote several wrong variations in the process).
Those lemmas are not really meant to be used in such a generality, but they are meant
as common foundations for various lemmas with stronger assumptions.

A further remark about proving these kinds of lemmas is that informal accounts are
typically very sloppy about handling the fact that $s(x)$ could cross $a$ in the context
of \cref{lem:parametric_integral}. The formalised version requires care here.

One last ingredient in the proof of \Cref{prop:corrugation} is that any
continuous loop is bounded, and this holds uniformly with
respect to any parameter moving in a compact set. Our representation of
loops as periodic functions instead of as functions on the compact space $\S$ means
we must do some extra work to invoke this fact. We thus introduced $\RR/\ZZ$, denoted $\SS_1$ in the code,
and some glue to go back and forth between 1-periodic functions and functions on
$\SS_1$. However this glue is tightly encapsulated: we only use it to
prove that, given any \separated topological space $X$ and any compact set $K$
in $X$, every continuous function $f$ from $X \times \RR$ to a normed space
such that each $f(x, \cdot)$ is $1$-periodic is bounded on $K \times \RR$.
Yet again we benefitted from \mathlib's strong topology library, invoking results about
quotient maps and \separated quotient spaces.

\subsection{The inductive argument}%
\label{sub:the_inductive_argument}

The proof of \Cref{thm:gromov} repeatedly uses the corrugation operation to
improve the given formal solution in more and more directions. Stating this
precisely requires a refinement of the notion of being holonomic. Given a linear subspace
$E' \subset E$, we say that $(f, \varphi) : E \to F \times \L{E}{F}$ is $E'$-holonomic
at $x$ if $Df(x)$ and $\varphi(x)$ coincide on $E'$.

\begin{lemma}\label{lem:inductive_step}
  Let $\F$ be a formal solution of some open and ample $\Rel \subset J^1(E, F)$.
  Let $K_1 \subset E$ be a compact subset, and let $K_0$ be a compact subset of
  the interior of $K_1$. Let $C$ be a closed subset of $E$.
  Let $p$ be a dual pair on $E$ and let $E'$ be a linear subspace of $E$
  contained in $\ker \pi$.
  Let $\varepsilon$ be a positive real number.

  Assume that $\F$ is $E'$--holonomic near $K_0$, and holonomic near $C$.
  Then there is a homotopy $\F_t$ such that:
  \begin{enumerate}
    \item
      $\F_0 = \F$~,\selink{local/h_principle.lean\#L287}
    \item
      $\F_t$ is a formal solution of $\Rel$ for all $t$~,\selink{local/h_principle.lean\#L403}
    \item
      $\F_t(x) = \F(x)$ for all $t$ when $x$ is near $C$ or outside
      $K_1$~,\selink{local/h_principle.lean\#L298}\selink{local/h_principle.lean\#L320}
    \item
      $\|\pr_F\F_t(x) - \pr_F \F(x)\| \le \varepsilon$ for all $t$ and all $x$~,%
      \selink{local/h_principle.lean\#L342}
    \item
      $\F_1$ is $E' \oplus \mathbb{R}v$--holonomic near $K_0$.\selink{local/h_principle.lean\#L359}
  \end{enumerate}
\end{lemma}

\begin{proof}[Proof sketch]
  We denote the components of $\F$ by $f$ and $\varphi$.
  Since $\Rel$ is ample, \cref{prop:loops} applied to
  $g \co x \mapsto Df(x)v$, $\beta \co x \mapsto \varphi(x)v$,
  $\Omega_x = \Rel(\F(x), p)$, and $K = C \cap K_1$ gives us a smooth family of loops
  $\gamma \co E \times [0, 1] \times \SS^1 \to F$ such that, for all $x$:
  \begin{itemize}
    \item $\forall t\, s,\; \gamma^t_x(s) \in \Rel(\F(x), \pi, v)$,
    \item $\forall s,\; \gamma^0_x(s) = \varphi(x)v$,
    \item $\bar \gamma^1_x = Df(x)v$,
    \item if $x$ is near $C$, $\forall t\, s,\; \gamma^t_x(s) = \varphi(x)v$.
  \end{itemize}
  Let $\rho: E \to \mathbb{R}$ be a smooth cut-off function which equals one on
  a neighbourhood of $K_0$ and whose support is contained in $K_1$.

  Let $N$ be a positive real
  and set $\F_t(x) = \big(f_t(x), \varphi_t(x)\big)$ where:
  \[
    f_t(x) = f(x) + t\rho(x)\Corr{p}{N}{\gamma_t},
  \]
  and:
  \[
    \varphi_t(x) = \Upd{p}{\varphi(x)}{\,\gamma^{t\rho(x)}_x(N\pi(x))} +
      \Rem{p}{N}{\gamma^1}.
  \]
  One then checks this homotopy is suitable using \Cref{prop:corrugation}
\end{proof}

Given \Cref{lem:parametricity_for_free}, the preceding lemma allows us to prove \Cref{thm:gromov}
by induction on a basis of directions in $E$. Specifically we choose a basis
$e \co \{1, \dots, n\} \to E$, take the dual basis $e^*$, and apply \Cref{lem:inductive_step}
$n$ times using the sequence of dual pairs $p_i = (e_i^*, e_i)$ and subspaces
$E_i = \Span(e_1, \dots, e_i)$.

One formalisation issue is that the whole construction carries around a lot of
data. On paper it is easy to state one lemma listing all this data once and
proving many properties. For us it was more convenient to give each property its
own lemma. Carrying around data, assumptions and constructions thus required some
planning. We mitigated this issue by using two ad-hoc structures which partly
bundle the data.

The \lstinline{landscape}\selink{local/h_principle.lean\#L72} structure records
three sets in a vector space, a
closed set \lstinline{C} and two nested compact sets \lstinline{K₀} and
\lstinline{K₁}. This is the ambient data for the local h-principle result. We
call this partly bundled because it doesn't include the data of the formal
solution we want to improve. Instead we have a Prop-valued structure
\lstinline{landscape.accepts}\selink{local/h_principle.lean\#L105} that takes a
landscape and a formal solution and
asserts some compatibility conditions. There are four conditions, which is
already enough motivation to introduce a structure instead of one definition
using the logical conjunction operator that would lead to awkward and error
prone access to the individual conditions.

The proof of this proposition involves an induction on a flag of subspaces
(nested subspaces of increasing dimensions). For the purpose of this induction
we use a second structure
\lstinline{step_landscape}\selink{local/h_principle.lean\#L91} that extends
\lstinline{landscape} with two more pieces of data, a subspace and a dual pair,
and a compatibility condition, namely the subspace has to be in the hyperplane
defined by the dual pair.

In this setup the loop family constructed by \Cref{prop:loops} is used to
construct a function whose arguments are some
\lstinline{(L : step_landscape E)}, a formal solution \lstinline{𝓕} and an
assumption \lstinline{(h : L.accepts R 𝓕)}.
Together with corrugation, it is used to build the homotopy of 1-jet sections
appearing in the proof of \Cref{lem:inductive_step} improving the formal
solution \lstinline{𝓕} in that step of the main inductive proof. A rather long
series of lemmas prove all the required properties of that homotopy,
corresponding to all conclusions of \Cref{lem:inductive_step}.

In the inductive construction itself, all conclusions are stated at once since
the induction requires knowing about each of them to proceed to the next step.
We could have introduced one more ad-hoc structure to record those conclusions
but this isn't needed (at least in that part of the project) since we need to
access its components only once.

We finish with a comment on induction in the context of Lean.
Since it is based on the calculus of inductive construction, Lean's foundations have
built-in support for inductive constructions. This can be used for instance to build
the addition on natural numbers. In this project we are talking about a
\emph{much} more involved inductive construction where each piece requires one to
prove many facts to proceed. In principle it would be possible to use an
inductive construction in the foundational sense but that would be extremely
cumbersome. Instead we state some existential statement and prove it by induction.

\section{Sphere eversion}%
\label{sec:sphere_eversion}

In this section we explain how to derive Smale's theorem from
\Cref{thm:gromov}. This is less direct than deriving it from the global version
of Gromov's theorem (for maps between manifolds), but still rather easy since
the source manifold $\SS^2$ has a simple model as a subset of a vector space and the target
manifold is just a vector space. In this section $E$ is a 3-dimensional
real vector space equipped with a scalar product and $\SS^2$ is the unit sphere in $E$.
For any point $x$ in $\SS^2$, the tangent space $T_x\SS^2$ to $\SS^2$ at $x$ is
the subspace $x^\perp$ of $E$ orthogonal to the line spanned by $x$. An immersion
of $\SS^2$ into $E$ is a smooth map $f$ defined near $\SS^2$ and such that
for every $x$ in $\SS^2$, $Df(x)$ is injective on $T_x\SS^2$. At face value
this may sound slightly stronger than the definition of an abstract immersion
from manifold theory. But one can easily prove that any abstract immersion extends to
an immersion in the elementary sense. In any case, using a definition of
immersion that is too strong would only make Smale's theorem stronger.

\begin{theorem}[{Smale \cite{Smale}}]
	There is a homotopy of immersions of $\SS^2$ into $E$ from the inclusion map to
	the antipodal map $a \co q \mapsto -q$.\selink{local/sphere_eversion.lean\#L426}
\end{theorem}

Because we want to deduce this from our statement about maps between
vector spaces, we need to be slightly careful. We denote by $B$ the ball
with radius $9/10$ in $E$. The relation we use is
\[
  \Rel = \left\{(x, y, \varphi) \in J^1(E, E) \;|\;
  x \notin B \implies \rst{\varphi}{x^\perp} \text{ is injective}\right\}.
\]
\begin{proposition}
  Any solution of $\Rel$ is an immersion of $\SS^2$ into $E$.\selink{local/sphere_eversion.lean\#L68}
  The relation $\Rel$ is open\selink{local/sphere_eversion.lean\#L131} and ample.\selink{local/sphere_eversion.lean\#L169}
\end{proposition}

As far as we know, this proposition is new. This makes sense because it is not
needed to deduce Smale's theorem from the global version of Gromov's theorem.
We will explain some ideas of the proof since it also provides examples
of geometrically-obvious facts whose proofs need some thought.

In order to prove that $\Rel$ is open, the main task is to fix $x_0 \notin B$
and $\varphi_0 \in \L{E, E}$ which is injective on $x_0^\perp$ and prove that,
for every $x$ close to $x_0$ and $\varphi$ close to $\varphi_0$, $\varphi$ is
injective on $x^\perp$. This is a typical situation where geometric intuition
makes it feel like there is nothing to prove.

One difficulty is that the subspace $x^\perp$ moves with $x$. We reduce to a fixed
subspace by considering the restriction to $x_0^\perp$ of the orthogonal
projection onto $x^\perp$. One can check this is an isomorphism as long as $x$
is not perpendicular to $x_0$.
More precisely, we consider $f \co J^1(E, E) \to \RR \times \L{x_0^\perp}{E}$ which sends
$(x, y, \varphi)$ to $(\langle x_0, x\rangle, \varphi \circ \pr_{x^\perp} \circ j_0)$
where $j_0$ is the inclusion of $x_0^\perp$ into $E$. The set $U$ of injective
linear maps is open in $\L{x_0^\perp}{E}$ and the map $f$ is continuous
hence the preimage of $\{0\}^c \times U$ is open. This is good enough for us because
injectivity of $\varphi \circ \pr_{x^\perp} \circ j_0$ implies injectivity of
$\varphi$ on the image of $\pr_{x^\perp} \circ j_0$ which is $x^\perp$ whenever
$\langle  x_0, x\rangle \neq 0$.

The next thing to prove is ampleness of $\Rel$. The key observation is that if
one fixes vector spaces $F$ and $F'$, a dual pair $(\pi, v)$ on $F$, and an
injective linear map $\varphi \co F \to F'$ then the updated map
$\Upd{p}{\varphi}{w}$ is injective if and only if $w$ is not in
$\varphi(\ker\pi)$. It then only remains to prove one last
intuitively obvious result: the complement of a line in a three dimensional space
is ample (we actually prove a more general result).\selink{local/ample_set.lean\#L237}
We have not been able to find any informal source that provides any explanation of this.

When using the global version of Gromov's theorem to prove Smale's theorem, the
preceding key observation is enough to deduce the ampleness of the relevant relation
from the ampleness of the complement of a subspace with codimension at least two.
In our case we still have some work to do. It suffices to prove that for
every $\sigma = (x, y, \varphi) \in \Rel$ and every dual pair $p = (\pi, v)$ on
$E$, the slice $\Rel(\sigma, p)$ is ample. If $x$ is in $B$ then $\Rel(\sigma,
p)$ is all of $E$ which is obviously ample. So we assume $x$ is not in $B$.
Since $\sigma$ is in $\Rel$, $\varphi$ is injective on $x^\perp$. The slice is the
set of $w$ such that $\Upd{p}{\varphi}{w}$ is injective on $x^\perp$.
Assume first $\ker\pi = x^\perp$.
Then $\Upd{p}{\varphi}{w}$ coincides with $\varphi$ on $x^\perp$ hence the
slice is all of $E$. Assume now that $\ker\pi \neq x^\perp$. The slice is
not very easy to picture in this case but one should remember that, up to
affine isomorphism, the slice depends only on $\ker \pi$. More precisely, if
we keep $\pi$ but change $v$ then the slice is simply translated in $E$. Here
we replace $v$ by the projection on $x^\perp$ of the vector dual to $\pi$ rescaled
to keep the property $\pi(v) = 1$. What has been gained is that we now have
$v \in x^\perp$ and $x^\perp = (x^\perp \cap \ker \pi) \oplus \RR v$. Since
$\varphi$ is injective on $x^\perp$, $\varphi(x^\perp \cap \ker \pi)$ is a line, and
$\Upd{p}{\varphi}{w}$ is injective on $x^\perp$ if and only if $w$ is in the
complement of this line according to the key observation above. So we are indeed back
to the fact that the complement of a line is ample.

Lastly, we need a homotopy of formal solutions of $\Rel$. Roughly, we want to use
$\F : (t, x) \mapsto (1-2tx, \rot{\pi t}{x})$ where $\rot{\alpha}{x}$ is the
rotation with angle $\alpha$ around the axis spanned by $x$. But some extra
care is needed to ensure smoothness near the origin
(this artificial difficulty is the price we must pay for
extending our domain to all of $E$).\selink{local/sphere_eversion.lean\#L328}

\section{Blueprint infrastructure}%
\label{sec:blueprint_infrastructure}

Before the formalisation started the first author wrote a detailed blueprint in
\LaTeX\ with all the definitions and lemmas that were expected to be required for
the proof. This was meant to prepare the formalisation work and allow
contributions from people who did not know the area of mathematics. This is a
well-known strategy, see \cite{Hales_blueprint} for an implementation at a much
bigger scale.

The new ingredient was to write and use a plugin for \textsf{plasTeX},\href{http://plastex.github.io/plastex/}{\link} a very
extensible TeX compiler written in python. This allows one to render the blueprint
document in HTML with hyperlinks to precise locations in Lean files corresponding to each result.
The software also produces a dependency graph that shows the progress of the project and
assists with coordination. This graph is based on manual dependency declarations so
that we can indicate a dependency even in the absence of a \LaTeX\ reference.
Each node of the graph is either a definition or a lemma statement and is colour-coded
to indicate whether something is stated or proven or ready to be stated or proven.
This \textsf{leanblueprint} plugin\href{https://github.com/PatrickMassot/leanblueprint}{\link} is now used by at least half a dozen
formalisation projects. Adapting it to work with other proof assistants would
be very easy.

During the project we continuously edited the blueprint text to include more
explanation that we found to be necessary, or to cater for minor changes in strategy. The
end result is that we now have a somewhat bilingual informal/formal account of
all our results, in which each side is useful for illuminating the other.

\section{Conclusion and future work}%
\label{sec:conclusion_and_future_work}

We believe that this work demonstrates that arguments in differential
topology are not beyond the reach of formalisation. There were indeed many
places where informal sources do not provide any explanation, except perhaps a
a picture. Providing formal proofs was a rather pleasant process overall.

In order to see what was gained, one can return to the proof sketch we wrote
for \Cref{lem:inductive_step}. It is typical of many proofs in
differential topology. We rather carefully described a construction, with a
quite a lot of input data, and then the reader is expected to agree that this
data satisfies the desired conditions. Here we think the main benefit of a
formalised version is that the input, assumptions, and desired output are very
clearly stated. Reaching that level of clarity is difficult to achieve without
a proof assistant. After seeing such a statement, most readers probably
still prefer to work out a mental picture rather than seeing details of the
proof, but they are in a much better position to do so. And of course having
such a precise statement doesn't prevent anyone from also writing a more vague
but less intimidating version first.

The next step in this project is to deduce from \Cref{thm:gromov} the global
version for maps between smooth manifolds. In informal accounts this seldom
occupies more than a couple of paragraphs. Of course this assumes the
theory of smooth manifolds is known, including jet spaces. We already have
many pieces in \mathlib but reducing to vector spaces actually requires some care.

In the more distant future, we hope that formalised mathematics will allow people to
engage with much of differential topology at many different scales, from
rough heuristic pictures to full details according to the reader's wishes.

%% file: acks.tex
We would like to thank the \mathlib\ community for developing a usable library with the
requisite material used in this project.
There are too many contributors that developed the parts of the library that we used to mention here explicitly, but we want to specifically thank a few of the contributors.
We would like to thank
Yury Kudryashov and Sébastien Gouëzel for developing the main part of the library for
calculus and integration.
We also want to thank Yury for the development of partitions of unity.
We greatly appreciate Heather Macbeth's work for writing the theory of rotations
specifically for this project.
We appreciate Anatole Dedecker's contributions to the theory of paths and ample sets.
We thank Johan Commelin and Scott Morrison for the first version of Carathéodory's theorem.
Special thanks to Gabriel Ebner for help with the \lstinline{has_uncurry} type-class.

The second author appreciates the support by Fondation Mathématique Jacques Hadamard
for the support via a postdoc fellowship.

%% file: arxiv.bbl
\newcommand{\etalchar}[1]{$^{#1}$}
\providecommand{\bysame}{\leavevmode\hbox to3em{\hrulefill}\thinspace}
\providecommand{\MR}{\relax\ifhmode\unskip\space\fi MR }
\providecommand{\MRhref}[2]{%
  \href{http://www.ams.org/mathscinet-getitem?mr=#1}{#2}
}
\providecommand{\href}[2]{#2}
\begin{thebibliography}{dMKA{\etalchar{+}}15}

\bibitem[Com20]{mathlib-paper}
The~{mathlib} Community, \emph{The {Lean} mathematical library}, Proceedings of
  the 9th ACM SIGPLAN International Conference on Certified Programs and Proofs
  (New York, NY, USA), CPP 2020, Association for Computing Machinery, 2020,
  p.~367–381, \href {http://dx.doi.org/10.1145/3372885.3373824}
  {\path{doi:10.1145/3372885.3373824}}.

\bibitem[CP88]{cic}
Thierry Coquand and Christine Paulin, \emph{Inductively defined types},
  COLOG-88, International Conference on Computer Logic, Tallinn, USSR, December
  1988, Proceedings, 1988, pp.~50--66.

\bibitem[dMKA{\etalchar{+}}15]{lean}
Leonardo~Mendon{\c{c}}a de~Moura, Soonho Kong, Jeremy Avigad, Floris van Doorn,
  and Jakob von Raumer, \emph{The {L}ean {T}heorem {P}rover ({S}ystem
  {D}escription)}, Automated Deduction - {CADE-25} - 25th International
  Conference on Automated Deduction, Berlin, Germany, August 1-7, 2015,
  Proceedings, 2015, pp.~378--388.

\bibitem[Gro86]{Gromov_PDR}
Mikhael Gromov, \emph{Partial differential relations}, Ergebnisse der
  Mathematik und ihrer Grenzgebiete (3) [Results in Mathematics and Related
  Areas (3)], vol.~9, Springer-Verlag, Berlin, 1986, \href
  {http://dx.doi.org/10.1007/978-3-662-02267-2}
  {\path{doi:10.1007/978-3-662-02267-2}},
  \url{https://doi.org/10.1007/978-3-662-02267-2}. \MR{864505}

\bibitem[Hal07]{Hales07}
Thomas~C. Hales, \emph{The jordan curve theorem, formally and informally}, Am.
  Math. Mon. \textbf{114} (2007), no.~10, 882--894,
  \url{http://www.jstor.org/stable/27642361}.

\bibitem[Hal12]{Hales_blueprint}
Thomas~C. Hales, \emph{Dense sphere packings}, London Mathematical Society
  Lecture Note Series, vol. 400, Cambridge University Press, Cambridge, 2012, A
  blueprint for formal proofs, \href
  {http://dx.doi.org/10.1017/CBO9781139193894}
  {\path{doi:10.1017/CBO9781139193894}},
  \url{https://doi.org/10.1017/CBO9781139193894}. \MR{3012355}

\bibitem[HIH13]{isabelle_topology}
Johannes H{\"{o}}lzl, Fabian Immler, and Brian Huffman, \emph{Type {C}lasses
  and {F}ilters for {M}athematical {A}nalysis in {I}sabelle/{HOL}}, Interactive
  Theorem Proving - 4th International Conference, {ITP} 2013, Rennes, France,
  July 22-26, 2013. Proceedings, 2013, pp.~279--294.

\bibitem[IT20]{Immler20}
Fabian Immler and Yong~Kiam Tan, \emph{The poincar{\'{e}}-bendixson theorem in
  isabelle/hol}, Proceedings of the 9th {ACM} {SIGPLAN} International
  Conference on Certified Programs and Proofs, {CPP} 2020, New Orleans, LA,
  USA, January 20-21, 2020 (Jasmin Blanchette and Catalin Hritcu, eds.), {ACM},
  2020, pp.~338--352, \href {http://dx.doi.org/10.1145/3372885.3373833}
  {\path{doi:10.1145/3372885.3373833}},
  \url{https://doi.org/10.1145/3372885.3373833}.

\bibitem[Lew19]{Lewis2019Hensel}
Robert~Y. Lewis, \emph{A formal proof of hensel's lemma over the p-adic
  integers}, Proceedings of the 8th ACM SIGPLAN International Conference on
  Certified Programs and Proofs (New York, NY, USA), CPP 2019, Association for
  Computing Machinery, 2019, p.~15–26, \href
  {http://dx.doi.org/10.1145/3293880.3294089}
  {\path{doi:10.1145/3293880.3294089}},
  \url{https://doi.org/10.1145/3293880.3294089}.

\bibitem[Sma58]{Smale}
Stephen Smale, \emph{A classification of immersions of the two-sphere}, Trans.
  Amer. Math. Soc. \textbf{90} (1958), 281--290, \href
  {http://dx.doi.org/10.2307/1993205} {\path{doi:10.2307/1993205}},
  \url{https://doi.org/10.2307/1993205}. \MR{104227}

\bibitem[Spr98]{Spring}
David Spring, \emph{Convex integration theory}, Monographs in Mathematics,
  vol.~92, Birkh\"{a}user Verlag, Basel, 1998, Solutions to the $h$-principle
  in geometry and topology, \href {http://dx.doi.org/10.1007/978-3-0348-0060-0}
  {\path{doi:10.1007/978-3-0348-0060-0}},
  \url{https://doi.org/10.1007/978-3-0348-0060-0}. \MR{1488424}

\bibitem[The22]{Theilliere22}
M\'{e}lanie Theilli\`ere, \emph{Convex integration theory without integration},
  Math. Z. \textbf{300} (2022), no.~3, 2737--2770, \href
  {http://dx.doi.org/10.1007/s00209-021-02785-9}
  {\path{doi:10.1007/s00209-021-02785-9}},
  \url{https://doi.org/10.1007/s00209-021-02785-9}. \MR{4381219}

\bibitem[vD21]{vandoorn2021haar}
Floris van Doorn, \emph{{Formalized Haar Measure}}, 12th International
  Conference on Interactive Theorem Proving (ITP 2021), Leibniz International
  Proceedings in Informatics (LIPIcs), vol. 193, 2021, pp.~18:1--18:17, \href
  {http://dx.doi.org/10.4230/LIPIcs.ITP.2021.18}
  {\path{doi:10.4230/LIPIcs.ITP.2021.18}}.

\bibitem[WS22]{Wieser2022}
Eric Wieser and Utensil Song, \emph{Formalizing geometric algebra in lean},
  Advances in Applied Clifford Algebras \textbf{32} (2022), no.~3, \href
  {http://dx.doi.org/10.1007/s00006-021-01164-1}
  {\path{doi:10.1007/s00006-021-01164-1}},
  \url{https://doi.org/10.1007/s00006-021-01164-1}.

\end{thebibliography}
